\documentclass[10pt,showpacs,floatfix,letterpaper,aps,pra,reprint, showkeys]{revtex4-1}
\pdfoutput=1
\usepackage[utf8]{inputenc}
\usepackage{caption, subcaption, graphicx}
\usepackage[usenames,dvipsnames]{color}
\usepackage{amsmath,amssymb,amsthm}  
\usepackage{dsfont}

\newtheorem{thm}{Theorem}
\newtheorem{lem}{Lemma}
\newtheorem{defin}{Definition}[section]

\newcommand{\mbb}{\mathbb}

\newcommand{\bra}[1]{\langle #1|}
\newcommand{\ket}[1]{|#1 \rangle}
\newcommand{\braket}[2]{\langle #1|#2\rangle}
\newcommand{\ketbra}[1]{\ket{#1}\bra{#1}}
\newcommand{\ketb}[2]{\ket{#1}\bra{#2}}
\newcommand{\ident}{\mathds{1}}

\newcommand\diag[1]{\mathrm{diag}\lb #1 \rb}

\def\ox{\otimes}

\def\lp{\left(}
\def\rp{\right)}
\def\ls{\left[}
\def\rs{\right]}
\def\lb{\left\{}
\def\rb{\right\}}

\def\a{\alpha}
\def\b{\beta}
\def\g{\gamma}
\def\d{\delta}

\def\s{\sigma}

 \def\O{\Omega}

\renewcommand\Re{\mathrm{Re}}
\renewcommand\Im{\mathrm{Im}}

\begin{document}
\title{Continuous decomposition of quantum measurements via qubit probe feedback}
\author{Jan Florjanczyk}
\author{Todd A. Brun}
\affiliation{Center for Quantum Information Science and Technology, \\
Communication Sciences Institute, Department of Electrical Engineering, \\
University of Southern California Los Angeles, CA 90089, USA.}
\date{\today}



\begin{abstract}
It is known that any two-outcome quantum measurement can be decomposed into a continuous stochastic process using a feedback loop. In this article, we characterize which of these decompositions are possible when each iteration of the feedback loop consists of a weak measurement caused by an interaction with a probe system. We restrict ourselves to the case when the probe is a qubit and the interaction Hamiltonian between the probe and system is constant. We find that even given the ability to perform arbitrary unitary pulses throughout the continuous decomposition, only generalized measurements with two distinct singular values are achievable. However, this is sufficient to decompose a generalized qubit measurement using a qubit probe and a simple interaction Hamiltonian.
\end{abstract}

\maketitle

\section{Introduction}

In~\cite{nonlocality} the authors describe a scheme for generalized quantum measurements that allows for the probabilities of each outcome to be monitored continuously. The ability to halt the quantum measurement at a desired confidence then plays a crucial role in proving the inability of two parties sharing Local Operations and Classical Communication (LOCC) to successfully identify orthogonal product states. The continuous scheme, however, relies on attaching a large number of ancilla and the ability to perform large unitary operations on their joint system. In~\cite{weakuniversal, generalizedstochastic} the authors develop an alternative scheme that instead uses only diffusive weak measurements~\cite{simplemodel} with closed-loop feedback. 

In this paper we aim to characterize a continuous measurement procedure that lies between the two schemes above.  We are motivated by both theoretical and experimental considerations.  Many quantum mechanical systems either have naturally slow measurement times, or can only be probed weakly.  For example, the beautiful experiments of Haroche and Raimond \cite{rydberg} use a stream of Rydberg atoms to repeatedly probe the state of a microwave mode in a superconducting cavity.  Homodyne and heterodyne measurements are widely used in optics, and produce a continuous output current.  Superconducting qubits (or similar solid-state devices) can be measured by a weak dispersive coupling to a microwave cavity, which can be measured in turn by homodyne measurement.  Magnetic resonance force microscopy \cite{MRFM} can do single-spin measurement by a continuous measurement procedure.  Moreover, latency is sufficiently low in modern experiments that it is possible to do continuous feedback in real time, as has already been demonstrated in the microwave cavity/Rydberg atom system \cite{rydberg-feedback}.

For most quantum systems, there is no direct way to implement a given generalized measurement.  Generalized measurements are performed by unitary coupling to an ancillary system, followed by projective measurement on the ancilla.  In this paper, we explore the type of measurements that can be built up from a particular type of fixed weak interaction.  Moreover, the type of protocols we explore in this paper are examples of closed-loop quantum control, where continuous measurement is fed back not just to control the Hamiltonian of the system but the continuous measurement itself.  Such feedback can be used, for example, to improve the accuracy of phase measurements \cite{Phase}.  Here, we use it to decompose generalized measurements, but no doubt other types of protocols can be done using similar techniques.

We consider a situation where the system to be measured can only be probed weakly, and the experimenter has only limited control over the system itself.  The interaction between the probe and the system is fixed.  The experimenter, however, has complete control over both the preparation and measurement of the probes.  We wish to characterize what combinations of probe states, measurements, and interaction Hamiltonians yield continuous decompositions of generalized measurements.

Our results will concern only qubit probes. Although this may seem like a rather narrow class of experiments it is a natural setting for two-outcome measurements. Our results will apply equally well to the procedure outlined in~\cite{nonlocality, weakuniversal} that allows for any general $n$-outcome measurement to be decomposed into a series of two-outcome measurements.

This paper is organized as follows: In Section~\ref{sec:theory} we introduce the discretized steps of a continuous measurement (weak measurements), and review how they can be constructed from a qubit probe. In Section~\ref{sec:result} we state and prove our main result about interaction Hamiltonians between the probe and system. In Section~\ref{sec:unitaries} we briefly discuss the role of unitary pulses in our scheme. In Section~\ref{sec:qubit} we exhibit how our model can be used to decompose a generalized diagonal measurement on a system qubit using a qubit probe and the interaction Hamiltonian $Z \ox Z$. We summarize these findings in Section~\ref{sec:conclusions}.

\section{\label{sec:theory}Reversible weak measurements and random walks}
In~\cite{weakuniversal, infinitesimal} the authors show how to to decompose an instantaneous quantum measurement into a continuous process. Such a decomposition must respect two properties: that the state of the system being measured evolve smoothly, and that the entire continuous process allow for multiple possible outcomes. If we consider a discretized version of the process, then a sequence of diffusive weak measurements~\cite{simplemodel} satisfies both of these requirements. A weak measurement is parametrized by a ``strength" parameter $\d$ and its associated operators have the general form
\[ M_k \propto \ident + \d \hat{\varepsilon}_k, \]
where $\hat{\varepsilon}_k$ is an operator of bounded norm. Since the outcome of each measurement is a random function of the state, a sequence of weak measurements forms a stochastic process.

However, if a stochastic process is to be a faithful decomposition of a quantum measurement, then the result of the process must depend only on the state being measured and not on the total time of the evolution. In~\cite{weakuniversal} the authors show that this can be accomplished by casting the weak measurement steps as corresponding to a $1$-dimensional random walk indexed by the pointer variable $x$. In this case the result of the process does not depend on total time but, instead, on the drift of the pointer $x$ which is, in turn, dependent on the state. Each weak measurement step updates $x$ to $x \pm \d$ depending on the result. Any dependence on the duration of the process is accounted for by constructing step operators $M_{\pm}(x)$ that cancel when applied in opposite directions. More precisely, the scheme requires the steps to be \emph{reversible} (Def. ~\ref{def:reversibility}).
\begin{defin}[Reversibility condition] \label{def:reversibility}
We say a one-parameter family of weak measurement operators $\{ M_{\pm}(x) \}_x$ satisfies the \emph{reversibility condition} if
\begin{equation}
	\label{eqn:reversibility}
	M_{\mp}(x\pm \d) M_{\pm}(x) \propto \ident
\end{equation}
for all $x$ in a given interval.
\end{defin}
The step operators $M_{\pm}(x)$ in Eq.~\eqref{eqn:reversibility} are chosen such that the first operator updates $x$ to $x \pm \d$ and the second returns it to $(x \pm \d) \mp \d = x$. If the product of the two is proportional to the identity, then the operators have no net effect on the system state up to a normalization constant.

\begin{figure}
	\includegraphics[width=1.0\linewidth]{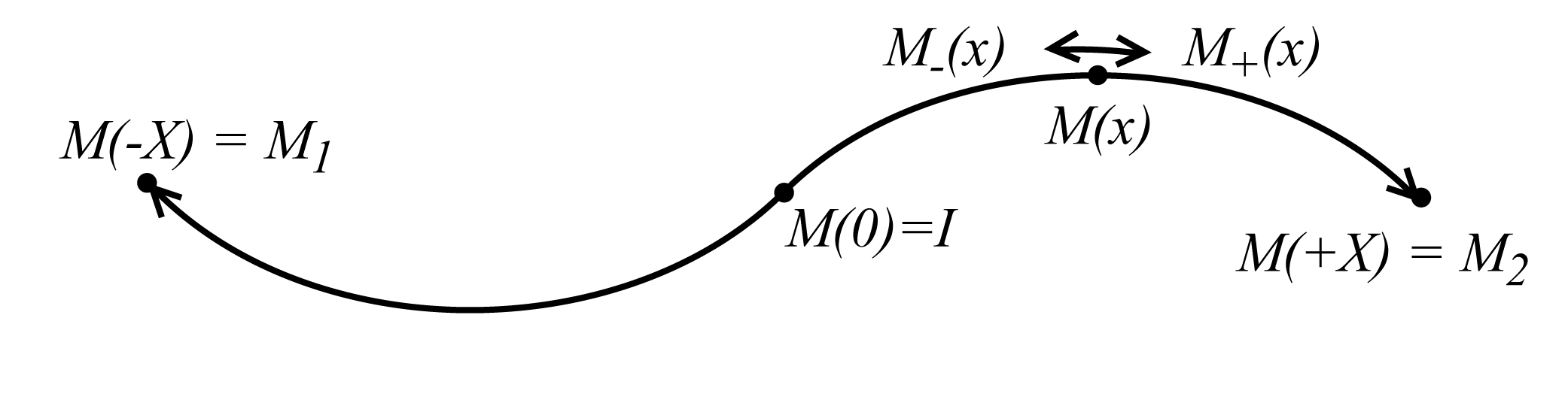}
	\caption{The random walk construction begins at $x=0$ with a weak measurement given by the step operators $M_{\pm}(0)$. Each successive step of the walk contributes to the total walk operator, Eq.~\eqref{eqn:totalwalkoperator}, via one of the two step operators and updates $x$ accordingly. The walk terminates at either endpoint $\pm X$ where the total walk operator is designed to match the desired instantaneous measurement.}
	\label{fig:graphicalmodel}
\end{figure}

To help clarify this construction, we provide a graphical representation in Figure~\ref{fig:graphicalmodel}. As one can see, there is one random walk performed by the pointer $x$ and another random walk performed by the evolution of the system state under each weak measurement. The path of both walks is uniquely parametrized by $x$.

The evolution operator $M(x)$ describing the total evolution of the state under the random walk above is given by the product of step operators from the initial state at $x=0$ to the current value of the pointer variable $x$,
\begin{equation}
	M(x) \propto \lb \begin{array}{lcr} \displaystyle\prod_{j=0}^{\lfloor |x| / \d \rfloor} M_+(j \d) & \hspace{0.25in} & x>0 \bigskip \\ \displaystyle\prod_{j=0}^{\lfloor |x| / \d \rfloor} M_-(-j \d) & & x<0 \end{array} \right.
	\label{eqn:totalwalkoperator}
\end{equation}

The particular instantaneous two-outcome measurement to which this decomposition corresponds is given by the endpoints of the random walk in the continuous limit,
\begin{equation}
	M_1 = \lim_{\d \rightarrow 0} M(X) \hspace{0.25in} \text{and} \hspace{0.25in} M_2 = \lim_{\d \rightarrow 0} M(-X) \label{eqn:matchingendpoints}
\end{equation}

Altogether, these operators define a \emph{continuous decomposition}.
\begin{defin}[Continuous decomposition] \label{def:continuousdecomposition}
We call a one-parameter family of weak measurements $\{ M_{\pm}(x) \}_x$ a \emph{continuous decomposition} of $\{ M_1, M_2 \}$ if $\{M_{\pm}(x)\}_x$ satisfy the reversibility condition (Def.~\ref{def:reversibility}) and the endpoints, as given by Eq.~\eqref{eqn:matchingendpoints}, match $M_1$ and $M_2$. We call $\{M_{\pm}(x)\}_x$ the \emph{step operators} of this decomposition.
\end{defin}

We are interested in performing the step operators described above via a projective measurement on a weakly interacting probe. In Figure~\ref{fig:circuit} we illustrate preparing a probe state, allowing the probe and system to interact for a short time $\d$, and then measuring the probe to update $x$. In the continuous limit, this feedback loop is considered to occur instantaneously. We call the circuit in Figure~\ref{fig:circuit} a \emph{probe feedback loop}.

\begin{defin}[Probe feedback loop]\label{def:probefeedbackloop} We say that a one-parameter family of $2$-outcome measurements $\{M_{\pm}(x)\}_x$ is generated by a \emph{probe feedback loop} if the measurement operators are of the form
\[ M_{\pm}(x) = \bra{\Phi^{\pm}(x)} e^{i \d H_{PS}} \ket{\s(x)}, \]
where $\ket{\s(x)}$ is the \emph{probe} qubit, $\ket{\Psi^{\pm}}$ are two orthogonal states associated with the destructive measurement of the \emph{detector}, and $H_{PS}$ is the interaction Hamiltonian between the probe $P$ and system $S$.
\end{defin}

\begin{figure}
\includegraphics[width=1.0\linewidth]{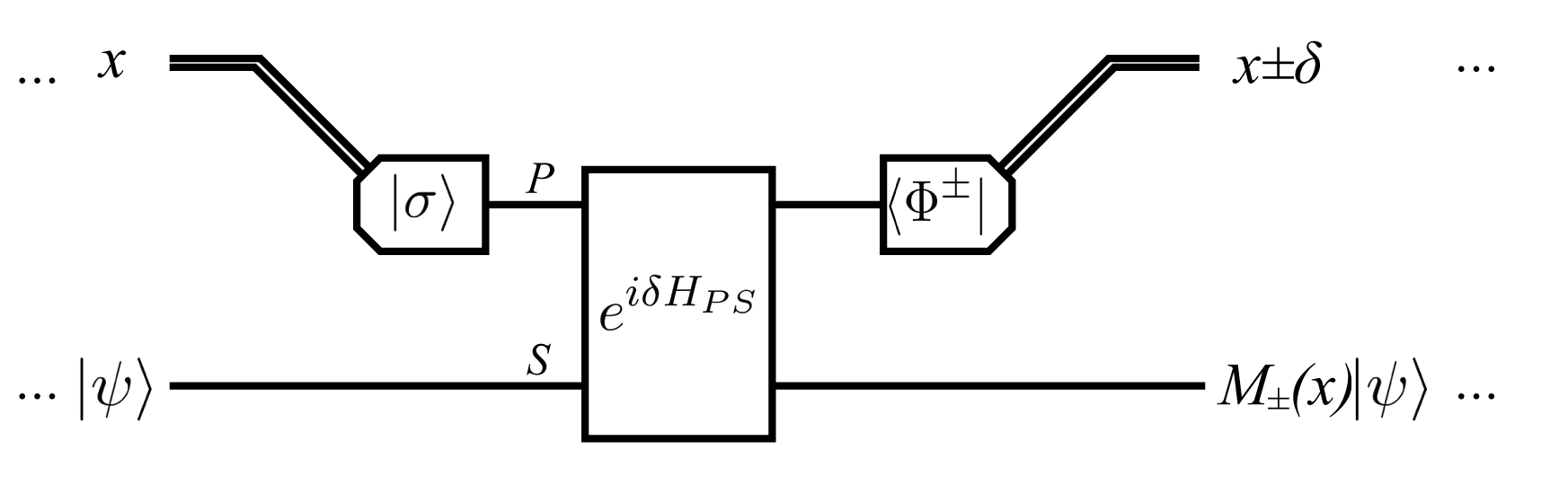}
\caption{In the above figure, we model a single step of the stochastic process for a pointer value $x$. Time flows from left to right. The value of the pointer variable is classical and thus symbolized by a double line. All single lines are quantum states. The probe state $\ket{\s (x)}_P$ is a qubit prepared according to the pointer variable. The probe and system $\ket{\psi}$ are then permitted to interact via the Hamiltonian $H_{PS}$. The probe is destroyed in the projective measurement $\{ \bra{\Phi^{\pm}(x)} \}$ which can depend on $x$ although we sometimes drop this dependence in our notation. The result of the measurement is then used to update the pointer variable for the following step.}
\label{fig:circuit}
\end{figure}

In some cases, the probe feedback loops are interlaced with weak unitary pulses on the system. These unitaries play a special role in that they encapsulate some notion of the experimentalist's power. In our efforts to characterize all measurements possible with a given $H_{PS}$, we will assume that an ``all-powerful" experimentalist will have the ability to perform any unitary pulse. We will examine constraints on the pulses available to the experimentalist in Section~\ref{sec:pulseconstraints}.

The most general Hamiltonian we can write for the interaction of a qubit probe and an arbitrary quantum system is the following
\begin{equation}
	H_{PS} = \ident \ox H_S + X \ox H_X + Y \ox H_Y + Z \ox H_Z
	\label{eqn:generalHamiltonian}
\end{equation}
where $\ident$, $X$, $Y$, $Z$ are the usual Pauli matrices on the probe state $P$, and $H_S$, $H_X$, $H_Y$, $H_Z$ are corresponding Hermitian matrices on the system $S$. Recalling our expression for the step operator above, we note that it now explicitly depends on the geometry of the probe state, the projective measurement, and the Pauli matrices. Combined with our requirement for step operator reversibility, Eq.~\eqref{eqn:reversibility}, we can begin to investigate which decompositions are possible.

Before we begin, however, we note that we are afforded one advantage through our random walk construction because the reversibility condition need not be met exactly. Consider that a classical random walk must take $O(N^2)$ steps to converge with fixed probability (where $N=X/\d$). This implies that the total walk operator in Eq.~\eqref{eqn:totalwalkoperator} will accumulate $N^2 O(\d)$ terms, $N^2 O(\d^2)$ terms, and so on. However, since the contribution of $N^2 O(\d^3)$ terms vanishes as $\d \rightarrow 0$ regardless of whether or not the step operators are exactly reversible, we only require that the reversibility condition of Eq.~\eqref{eqn:reversibility} be met only up to $O(\d^2)$. As we discover below, this requirement is still the source of the most stringent limitations on possible decompositions.

\section{\label{sec:result}Main result}

We consider Hamiltonians where the strength of various terms cannot be engineered. We call these ``fixed'' in the sense that at every weak-measurement step, the interaction between the probe and the system is identical. We allow for weak unitary pulses to be applied in between each step and we later specify which types of pulses are necessary to achieve the class of measurements we find below.

\begin{thm} \label{thm:MAIN}
Any continuous decomposition (Def.~\ref{def:continuousdecomposition}) with step operators generated by a probe feedback loop (Def.~\ref{def:probefeedbackloop}), can only match a $2$-outcome measurement of the form
\[M_1 = U_1 \lp \alpha \; \Pi_S + \beta \; \Pi_{S^{\perp}} \rp V, \]
\[M_2 = U_2 \lp \sqrt{1-\alpha^2} \; \Pi_S + \sqrt{1-\beta^2} \; \Pi_{S^{\perp}} \rp V, \]
where $U_1$, $U_2$, and $V$ are unitary matrices and $\Pi_S$, $\Pi_{S^{\perp}}$ are projectors onto orthogonal subspaces of the system space.
\end{thm}

\subsection{Proof of Theorem~\ref{thm:MAIN}}

We begin with the following observation about weak-measurements generated from qubit probes; when interpreted on the Bloch sphere, there is a geometric constraint between the probe and the detector.
\begin{defin}[Probe basis] \label{def:probebasis}
For any qubit probe $\ket{\s(x)}$ and projective qubit measurement $\bra{\Psi^{\pm}}$, we define a real orthonormal basis for the Bloch sphere $\{\vec{n}_1(x), \vec{n}_2(x), \vec{n}_3(x)\}$. We call this a \emph{probe basis} if 
\begin{itemize}
\item $\vec{\s}$, the Bloch vector associated with the probe state $\ket{\s(x)}$, is no further than distance $\d$ from $\vec{n}_1$, 
\item $\vec{n}_2$ is the Bloch vector associated with $\ket{\Phi^+}$, 
\item and $\vec{n}_3 = \vec{n}_1 \times \vec{n}_2$.
\end{itemize}
\end{defin}
The constraint itself is expressed in the following lemma, the proof of which can be found in the Appendix.
\begin{lem}[Probe basis of a weak measurement] \label{lem:probebasis}
Any diffusive weak measurement given by a probe feedback loop (Def.~\ref{def:probefeedbackloop}) with a probe basis $\{\vec{n}_1(x), \vec{n}_2(x), \vec{n}_3(x)\}$ must have $\vec{n}_2 \cdot \vec{\s} \sim O(\d)$. Thus, an orthonormal basis for the Bloch sphere that approximates the probe and detector always exists.
\end{lem}

The lemma above yields an important tool for our analysis. Recall that for a given interaction Hamiltonian we seek to characterize the weak measurement step operators achievable via \emph{any} probe feedback loop parametrized by $x$. However, we can instead fix a probe basis for the probe feedback loop and consider a family of interaction Hamiltonians $H'_{PS}(x)$ which give rise to the same set of weak measurement step operators. This transformation is performed with the following identifications
\begin{eqnarray*}
	H_{PS} & = & X \ox H_X + Y \ox H_Y + Z \ox H_Z \\
	& = & \vec{n}_1(x) \cdot \vec{P} H_1(x) + \vec{n}_2(x) \cdot \vec{P} H_2(x) + \vec{n}_3(x) \cdot \vec{P} H_3(x)
\end{eqnarray*}
where $\vec{P} = \ls X \; , \; Y \; , \; Z \rs^T$ and
\[ \ls \begin{array}{c} H_1(x) \\ H_2(x) \\ H_3(x) \end{array} \rs = \ls \begin{array}{c|c|c} & & \\ \vec{n}_1(x) & \vec{n}_2(x) & \vec{n}_3(x) \\ & &  \end{array} \rs \ls \begin{array}{c} H_X \\ H_Y \\ H_Z \end{array} \rs . \]
We will also later abbreviate the vectors above as $\vec{H}'(x) = \ls H_1(x) \; , \; H_2(x) \; , \; H_3(x) \rs^T$ and $\vec{H} = \ls H_X \; , \; H_Y \; , \; H_Z \rs^T$. Finally, we define an interaction Hamiltonian in the probe basis
\[ H'_{PS}(x) = X \ox H_1(x) + Y \ox H_2(x) + Z \ox H_3(x), \]
and this yields an advantageous rewriting of the weak measurement step operators
\begin{eqnarray*}
	M_{\pm}(x) & = & \bra{\Psi^{\pm}(x)} \exp \lp i \d H_{PS}\rp \ket{\s(x)} \\
	& = & \bra{\pm} \exp \lp i \d H'_{PS}(x) \rp \ket{0} + O(\d).
\end{eqnarray*}
In this basis, the detector states $\ket{\Phi^{\pm}}$ are $\ket{\pm}$, the $\pm1$ eigenstates of $X$, and the initial state $\ket{\s(x)}$ is close to $\ket{0}$, the $+1$ eigenstate of $Z$. This choice also allows us to ignore the $\ident_P \ox H_S$ term in the general Hamiltonian Eq.~\eqref{eqn:generalHamiltonian} since $ \bra{\pm} \ident \ket{0} = \bra{\pm} Z \ket{0}$ and any contribution from $H_S$ can be rewritten as part of $H_3$.

In lemma~\ref{lem:probebasis} we required that the probe and detector Bloch vectors be orthogonal only up to $O(\d)$. We will therefore allow the probe states to be perturbed from $\ket{0}$ in our analysis. This causes an adjustment in our expression for the weak measurement step operators, parametrized by two functions $c(x)$ and $\psi(x)$,
\begin{eqnarray*}
	\ket{\s(x)} & = & \cos \lp \d c(x) \rp \ket{0} + \sin \lp \d c(x) \rp e^{i \psi (x)} \ket{1} \\
	& \approx & \ket{0} + \d c(x) e^{i \psi(x)} \ket{1} \\
	& = & \ket{0} + \d \ket{\Delta (x)},
\end{eqnarray*}
where in the last line we've implicitly defined $\ket{\Delta (x)} = c(x) e^{i \psi(x)} \ket{1}$. Equivalently, this contributes a term of $O(\d)$ to our step operators
\begin{equation}
	M_{\pm}(x) = \bra{\pm} e^{i \d H'_{PS}(x)} \ket{0} + \d \bra{\pm} e^{i \d H'_{PS}(x)} \ket{\Delta(x)} .
\end{equation}
Grouping together $O(1)$, $O(\d)$ and $O(\d^2)$ terms in the above expression yields
\begin{eqnarray*}
	\lefteqn{M_{\pm}(x) = \bra{\pm} \Big( \ket{0} + \d \ket{\Delta(x)} \Big) \cdot \ident} \\
	& & + i \d  \bra{\pm} H'_{PS}(x) \Big( \ket{0} + \d \ket{\Delta(x)} \Big) - \frac{\d^2}{2} \bra{\pm} H'_{PS}(x)^2 \ket{0} \\
	& = & \frac{\ident}{\sqrt{2}} + \d \lp i \bra{\pm}H'_{PS}(x) \ket{0} \pm \frac{c(x) e^{i \psi(x)} \ident}{\sqrt{2}} \rp \\
	& &  - \frac{\d^2}{2} \Big( \bra{\pm} H'_{PS}(x)^2 \ket{0} - 2i \bra{\pm}H'_{PS}(x) \ket{\Delta(x)} \Big) \\
	& = & \frac{\ident}{\sqrt{2}} + \d M^{(1)}_{\pm}(x) - \frac{\d^2}{2} M^{(2)}_{\pm}(x),
\end{eqnarray*}
where we've implicitly defined $M^{(1)}_{\pm}(x)$ and $M^{(2)}_{\pm}(x)$ to collect the $O(\d)$ and $O(\d^2)$ terms. We can now write the reversibility condition in terms of the above:
\begin{eqnarray}
	\lefteqn{M_{\mp}(x \pm \d) M_{\pm}(x)} \nonumber \\
	& = & \lp \frac{\ident}{\sqrt{2}} + \d M^{(1)}_{\mp}(x \pm \d) - \frac{\d^2}{2} M^{(2)}_{\mp}(x \pm \d) \rp \nonumber \\
	& & \cdot \lp \frac{\ident}{\sqrt{2}} + \d M^{(1)}_{\pm}(x) - \frac{\d^2}{2} M^{(2)}_{\pm}(x) \rp \nonumber \\
	& = & \frac{\ident}{2} + \frac{\d}{\sqrt{2}} \lp M^{(1)}_{\mp}+ M^{(1)}_{\pm} \rp \label{eqn:O2toappendix} \\
	& &  - \frac{\d^2}{2\sqrt{2}} \lp \mp 2 \partial_x M^{(1)}_{\mp}+ M^{(2)}_{\mp} + M^{(2)}_{\pm} - 2\sqrt{2}M^{(1)}_{\mp}M^{(1)}_{\pm} \rp, \nonumber
\end{eqnarray}
where we've dropped $x$-dependence in the last line for legibility. First, for the $O(\d)$ term we find
\[ M^{(1)}_{\mp}+ M^{(1)}_{\pm} = i \sqrt{2} \langle H'_{PS} \rangle_0 = i \sqrt{2} H_3. \]
We provide the calculations for the following $O(\d^2)$ terms in Appendix~\ref{sec:miscellaneouscalculations}:
\begin{eqnarray}
	\lefteqn{M^{(2)}_{\mp}+ M^{(2)}_{\pm} - 2\sqrt{2}M^{(1)}_{\mp}M^{(1)}_{\pm}} \nonumber \\
	& = & \sqrt{2} \lp 2 (H_2^2 +H_3^2) + i \ls H_1 \pm H_3, H_2 \rs \right. \nonumber \\
	& &  \left. - i \lb H_1, H_2 \rb \pm \ls H_3, H_1 \rs + 4ce^{i \psi} H_2 \rp . \label{eqn:d2terms}
\end{eqnarray}
Expanding the $\partial_x M^{(1)}_{\mp}$ term is a bit more complicated since it corresponds to an infinitesimal rotation of the probe basis at each value of $x$. We can define an axis of rotation on the Bloch sphere with three components $\vec{\O}(x) = \ls \O_1(x)\ \; , \; \O_2(x) \; , \; \O_3(x) \rs$, so that $\partial_x \vec{H}'(x) = \vec{\O}(x) \times \vec{H}'(x)$. This implies that
 \[ \mp 2 \partial_x M^{(1)}_{\mp} = \sqrt{2} \ls i, -1, \mp i \rs \cdot \lp \vec{\O} \times \vec{H}' \rp, \]
where we've ignored any term proportional to the identity operator as these automatically satisfy the reversibility condition. Altogether these reductions yield the expression of interest for the reversibility condition,
\begin{eqnarray*}
    \lefteqn{M_{\mp}(x \pm \d) M_{\pm}(x)} \\
    & = & \frac{\ident}{2} + i \d H_3 - \frac{\d^2}{2} \Big\{ \ls i, -1, \mp i \rs \cdot \lp \vec{\O} \times \vec{H}' \rp \\
    & & + 2 (H_2^2 +H_3^2) + i \ls H_1 \pm H_3, H_2 \rs  \\
    & & - i \lb H_1, H_2 \rb \pm \ls H_3, H_1 \rs + 4ce^{i \psi} H_2 \Big\} .
\end{eqnarray*}

Finally, we group terms into four types: constant-Hermitian $A$, stochastic-Hermitian $B$ (that is, with a factor of $\pm$), constant-anti-Hermitian $i\bar{A}$, and stochastic-anti-Hermitian $i\bar{B}$:
\begin{equation}
	M_{\mp}(x \pm \d) M_{\pm}(x) = \frac{\ident}{2} + i \d H_3 - \frac{\d^2}{2} \lb A \pm B + i\bar{A} \pm i \bar{B} \rb, \label{eqn:revers}
\end{equation}
where
\begin{eqnarray*}
    A & = & - [ \vec{\O} \times \vec{H}' ]_2 + 4c \cos \psi H_2 + 2H_2^2 + 2H_3^2 + i \ls H_1, H_2 \rs, \\
    B & = & - i \ls H_2, H_3 \rs, \\
    \bar{A} & = & [ \vec{\O} \times \vec{H}' ]_1 + 4c \sin \psi H_2- \lb H_1, H_2 \rb, \\
    \bar{B} & = & -[ \vec{\O} \times \vec{H}' ]_3 - i\ls H_1(x), H_3(x) \rs.
\end{eqnarray*}
If the reversibility condition is to be satisfied then these, along with the $O(\d)$ term, must each be individually proportional to $\ident$. This can be done either by restrictions on the Hamiltonian terms, or by canceling the terms through unitary pulses applied when the random walk changes direction.

To eliminate the $H_3$ term, we must either set $H_3 \propto \ident$, or perform a weak unitary pulse of the form $U_1 = \exp \lp 2i\d H_3\rp$. As it happens, setting $H_3 \propto \ident$ does not change the analysis that follows and thus we'll assume instead that the experimentalist performs the pulse $U_1$ at each reversal of the walk direction.

Next, we assume that both the $\bar{A}$ and $\bar{B}$ term can be eliminated via a series of weak unitary pulses of the form $U_2 = \exp \lp i \d^2 H \rp$ where $H$ is some Hermitian operator containing linear combinations and products of $H_X$, $H_Y$, and $H_Z$. This leaves only $A$ and $B$ terms:
\begin{eqnarray}
    A & = & - [ \vec{\O} \times \vec{H}' ]_2 + 4c \cos \psi H_2 + 2H_2^2 + i \ls H_1, H_2 \rs \label{eqn:A} \\
    B & = & - i \ls H_2, H_3 \rs \label{eqn:B} .
\end{eqnarray}
Since $B$ is traceless it cannot be proportional to the identity and must be set to $0$. However, $A$ is not traceless and we must consider a more complicated solution, one where $A$ is equal to $\a \ident$  for some constant $\a$:
\[ \O_3H_1 - \O_1H_3 + 4c \cos \psi H_2 + 2H_2^2 + i \ls H_1, H_2 \rs = \a \ident . \]
Using lemma~\ref{lem:commutatoridentity} in the appendix we find that 
\[ \O_3H_1 - \O_1H_3 + 4c \cos \psi H_2 + 2H_2^2 = \a \ident \]
also implies that $\ls H_1, H_2 \rs = 0$. Together with the commutation relation from $B$, this means that we can now express all Hamiltonian terms in one common diagonal basis:
\begin{equation}
	\label{eqn:allHPS}
	H_{PS} = \sum_j \lp Xx_j(x) + Yy_j(x)+ Zz_j(x) \rp \ketbra{j(x)} .
\end{equation}
In order to satisfy the condition $A \propto \ident$, we must consider the diagonal components of $H_2$ as they appear in equation Eq.~\eqref{eqn:A}:
\begin{equation}
	\label{eqn:riccati}
	\partial_x y_j(x)  = q_0(x) + q_1(x) y_j(x) +  q_2(x) y_j^2(x),
\end{equation}
where $q_0(x)=\a$ from above, $q_1(x) = -4 c(x) \cos \psi(x)$, and $q_2(x) = 2$. This differential equation is a special instance of the Riccati equation, the solution to which can be found in~\cite{murphyODEs}. In particular, if any solution, $y^{(1)}(x)$, is known, then the general solution is of the form
\[ y_j(x) = y^{(1)}(x) + \frac{\Phi(x)}{C_j - \int q_0(x) \Phi(x) dx} .\]
for $\Phi(x) = \exp \int 2q_0(x)y^{(1)}(x) + q_1(x) dx$. The important feature of this solution is that there is only one free parameter $C_j$ available to match any boundary condition. 

To complete the proof, we focus on which instantaneous measurements $M_1$, $M_2$ are achievable at the end points of a continuous decomposition. First, note that all terms in Eq.~\eqref{eqn:allHPS}, including the diagonal basis $\ketbra{j(x)}$, are assumed to be $x$-dependent. Consider the unitary $U(x)$ which diagonalizes $H_1(x)$, $H_2(x)$, and $H_3(x)$. Each of these is a linear combination of $H_X$, $H_Y$, and $H_Z$, and since they are all linearly independent, $U(x)$ must also diagonalize $H_1$, $H_2$, and $H_3$. Whatever unitary does \emph{this}, however, cannot depend on $x$, and therefore, the basis $\ketbra{j}$ is \emph{not} $x$-dependent. Only the $y_j(x)$, $x_j(x)$, and $z_j(x)$ coefficients depend on $x$. This means that every step operator is diagonal in the same basis, and we can write the general form
\begin{eqnarray*}
	M_{\pm}(x) & = & \frac{\ident}{\sqrt{2}} \mp \frac{\d}{\sqrt{2}} \sum_j \lp y_j (x) - c(x) \cos \psi(x) \rp \ketbra{j}\\
	& &  + \frac{\d}{\sqrt{2}} \sum_j i \lp z_j(x) \pm x_j(x) \pm c(x) \sin \psi(x) \rp \ketbra{j}.
\end{eqnarray*}
Thus the endpoint measurement operators must also be diagonal in the $j$ basis, and the first of these has the form
\begin{eqnarray*}
	M_1 & \propto & \underset{\d \rightarrow \infty}{\lim} \prod_{j=0}^{\lfloor X/\d \rfloor} M_+ (j \d) \\
	& \propto & \underset{\d \rightarrow \infty}{\lim} \prod_{j=0}^{\lfloor X/\d \rfloor} \diag{1 - \d y_j(x) + i \d \lp z_j(x) + x_j(x) \rp} \\
	& = & \diag{\exp \lp \int_0^X -y_j(x) + i \lp z_j(x) + x_j(x) \rp dx \rp}.
\end{eqnarray*}
where the notation $\diag{\cdot}$ represents a diagonal matrix with entries indexed by $j$. Both $x_j(x)$ and $z_j(x)$ only contribute a total phase to each of the diagonal elements. If we let $w_j^{(1)} = \int_0^X z_j(x) + x_j(x) dx$ and $W_1 = \diag{\exp i w_j^{(1)}}$ then
\[ M_1 \propto W_1 \cdot \diag{\exp \lp - \int_0^X y_j(x) dx \rp} .\]
Following a similar procedure, we find that 
\[ M_2 \propto W_2 \cdot \diag{\exp \lp \int_{-X}^0 y_j(x) dx \rp} \]
with $W_2$ defined accordingly. Recall however, that these must form a complete measurement, and so they must satisfy $M_1^{\dag}M_1 + M_2^{\dag}M_2 = \ident$. This condition restricts the parameter $C_j$. Consider the $j^{\text{th}}$ diagonal entry of $M_1^{\dag}M_1$ (up to an overall normalization identical for all $j$),
\begin{eqnarray*}
	\lefteqn{\exp \lp - 2 Y(X) - 2 \int_0^X \frac{\Phi(x)}{C_j - \int q_0(x) \Phi(x) dx} dx \rp} \\
	& = & \exp \lp - 2 Y(X) - 2 \int_{\Phi^{-1}(0)}^{\Phi^{-1}(X)} \frac{d \Phi}{C_j - \int q_0(\Phi^{-1}(x)) d\Phi} \rp \\
	& = & e^{-2Y(X)} \lp \frac{C_j - Q(0)}{C_j - Q(X)} \rp^2,
\end{eqnarray*}
where we've implicitly defined
\[ Y(X) = \int_0^X y^{(1)}(x) dx \hspace{0.25in} Q(X) = \int q_0(x) \Phi(x) dx \]
It is not difficult to see, then, that the expression for each diagonal element is quadratic in $C_j$ and has only two solutions. In fact, even if one were to consider two arbitrary boundaries for the random walk, i.e.: that $-X$ be replaced with $X_1$ and $+X$ be replaced with $X_2$, then the expression for $C_j$ would still be quadratic and, again, yield only two possible solutions.

Finally, we can group the diagonal elements of $M_1$ and $M_2$ in terms of the two possible values of $C_j$ and express the entire measurement operators as the linear combination of two orthogonal projectors. In order to complete the proof of Theorem~\ref{thm:MAIN}, we must justify the appearance of unitaries $U_1$, $U_2$ and $V$. First, the rotation $V$ is simply any rotation that an experimentalist applies to the system $S$ before beginning the measurement procedure. For this reason, it does not depend on the measurement outcome. On the other hand, $U_1$ and $U_2$ are rotations applied to the system \emph{after} the continuous measurement procedure is complete and do not have to be identical. We can absorb the diagonal unitary matrices $W_1$ and $W_2$, which accumulate over the continuous procedure, into the definitions of $U_1$ and $U_2$ respectively.

\subsection{\label{sec:pulseconstraints}Constraints on weak unitary pulses}
In the section above we made repeated use of weak unitary pulses to reduce the equations for reversibility. However, some of these were generated by Hamiltonians with products of $H_1$, $H_2$, and $H_3$. Since we aim for our result to apply for a general but fixed interaction Hamiltonian $H_{PS}$ or, equivalently, three system Hamiltonian terms $H_X$, $H_Y$, $H_Z$, it may be too demanding to assume that a set of unitary pulses generated by their products would also be readily available. For this reason, we now consider satisfying the reversibility condition again, but only allowing weak unitary pulses generated by linear combinations of $H_X$, $H_Y$, and $H_Z$.

We examine this more restricted set of solutions to the reversibility condition by reintroducing the constraints on $H_1$, $H_2$, and $H_3$ that we removed in the previous section by using a weak unitary pulse generated by their products. In particular, we reintroduce the conditions on $\bar{A}$ and $\bar{B}$ that required
\[ \lb H_1, H_2 \rb \propto \ident, \hspace{0.5in} \text{and} \hspace{0.5in} \ls H_1, H_3 \rs \propto \ident. \]
The second of these is already automatically satisfied. The first yields the following relationship between the eigenvalues of $H_1$ and $H_2$:
\begin{equation}
	x_j(x)y_j(x) = \g(x) \hspace{0.15in} \forall j,x,
	\label{eqn:linearpulseconstraint}
\end{equation}
for some $\g(x)$ independent of $j$.

So far our result has only placed a restriction on the number of distinct singular values that the measurement can have. Here we'll actually be able to prove something about the singular values in the interaction Hamiltonian. We define $\vec{\lambda}_j = (\lambda_j^{(x)}, \lambda_j^{(y)}, \lambda_j^{(z)})$ as the triplet of the $j^{th}$ eigenvalues of $H_X$, $H_Y$, and $H_Z$, and $\vec{\lambda}'_j(x) = \lp x_j(x), y_j(x), z_j(x) \rp$ as the triplet of the $j^{th}$ eigenvalues of $H_1(x)$, $H_2(x)$, and $H_3(x)$. In the solution from the previous section, each $\vec{\lambda}'_j(x)$ corresponds to some $\vec{\lambda}_j$ via a rotation. This rotation takes each triplet from the original basis, to the probe basis in the same way that $\ket{\s(x)}$ and $\bra{\Psi^{\pm}(x)}$ were rotated from the original basis to $\ket{0}$ and $\bra{\pm}$.

Since we found only two solutions for $y_j(x)$, we must also restrict $\vec{\lambda}_j$ to lie in one of two planes in the original basis. Furthermore, the restriction in Eq.~\eqref{eqn:linearpulseconstraint} requires that the vectors $\vec{\lambda}'_j(x)$ be constrained to lie on one of two lines parellel to $z$ in the probe basis. In the original basis, this restricts all vectors $\vec{\lambda}_j$ to lie on one of two parallel lines, $l_1$ and $l_2$.

If we have only two assignments for $\vec{\lambda}_j$ then any probe basis is possible so long as the resulting $y_j(x)$ match those of the solution to the Ricatti equation above. However, with three or more assignments, we can only allow probe bases related by a rotation around the axis parallel to the lines $l_1$,$l_2$. The constraint Eq.~\eqref{eqn:linearpulseconstraint} limits this even further and we require that $y_1(x) = - y_2(x)$, meaning also that $\vec{\lambda}_j = c_j \vec{\lambda}_0 \pm \vec{\lambda}_1$ for some constant $c_j$, $\vec{\lambda}_0$ parallel to the lines $l_1$, $l_2$, and arbitrary $\vec{\lambda}_1$. This last expression for $\vec{\lambda}_j$ completely restricts the interaction Hamiltonian one should use if only pulses generated by linear combination of Hamiltonian terms are available.

\section{\label{sec:unitaries}Interleaving unitaries}
In~\cite{weakuniversal} the authors generalize their result for positive measurement operators to general measurement operators by taking the polar decomposition of the endpoint measurement operators. In other words, $M_1 = V_1 ( M_1^{\dag}M_1)^{1/2}$, and similarly for $M_2$. They then construct a one-parameter family of unitary operations $\{ V(x) \}_x$ that yield $V_1$ at $x=X$, $V_2$ at $x=-X$ and $\ident$ at $x=0$. The step operators are first constructed so as to correspond to the positive operators $(M_{1,2}^{\dag}M_{1,2})^{1/2}$ and padded by unitary operators chosen from the family $V(x)$ as follows
\begin{equation}
	M_{\pm}(x) = V(x \pm \d) \tilde{M}_{\pm}(x) V^{\dag}(x)
	\label{eqn:stepwithpadding}
\end{equation}
where $\tilde{M}_{\pm}(x)$ is the step operator for the positive part of the polar decomposition. We will show that in our analysis, padding the step operator with this family of unitary operators is equivalent to a shift in the $H_1(x)$ term.
Consider expanding the $\tilde{M}_{\pm}(x)$ term in Eq.~\eqref{eqn:stepwithpadding} in terms of $\d$,
\begin{eqnarray*}
	M_{\pm}(x) & = & \frac{1}{\sqrt{2}} V(x \pm \d)V^{\dag}(x) \\
	& & + i \d V(x \pm \d) \bra{\pm} H'_{PS}(x) \ket{0} V^{\dag}(x) .
\end{eqnarray*}
Recall that $V(x)$ forms a \emph{continuous} family of unitary operators, and if we let $V(x) = e^{i \d G(x)}$, then 
\[ V(x \pm \d) V^{\dag}(x) = \ident \pm i \d \partial_x G(x) .\]
This means we can summarize Eq.~\eqref{eqn:stepwithpadding} as
\[ M_{\pm}(x) = \frac{\ident}{\sqrt{2}} +  \frac{\d}{\sqrt{2}} \lb \pm i \lp \tilde{H}_1+\partial_x G \rp \mp \tilde{H}_2 + i \tilde{H}_3 \rb \]
where $\tilde{H}_{1,2,3} = V H_{1,2,3} V^{\dag}$. Thus, for this scheme we still recover the result of Thm.~\ref{thm:MAIN} but with a small modification. Namely, an experimentalist now has the power to introduce a shift to the $H_1$ term which contributes directly to the unitary terms $U_{1,2}$ that appear in the endpoint measurements $M_{1,2}$. It is important to note, however, that we've not affected the singular value decomposition of $M_{1,2}$, for which there are still only two distinct singular values.

\section{\label{sec:qubit}General diagonal measurement of a qubit}
While the results of this paper restrict the class of measurements that can be achieved in general by this model, it is sufficient to realize any $2$-outcome measurement on a qubit. We now consider performing a generalized diagonal measurement on a qubit via a continuous decomposition (Def.~\ref{def:continuousdecomposition}). A generalized diagonal measurement takes the form
\[ M_1 = W_1 \cdot \ls \begin{array}{cc} \a&0 \\ 0&\b \end{array} \rs,  \hspace{0.25in} M_2 = W_2 \cdot \ls \begin{array}{cc} \sqrt{1-\a^2} & 0 \\ 0 & \sqrt{1-\b^2} \end{array} \rs \]
where $W_1$ and $W_2$ are unitary matrices. We will effectuate the continuous decomposition via a sequence of probe feedback loops (Def.~\ref{def:probefeedbackloop}) and the interaction Hamiltonian $Z_P \ox Z_S$. Expressing the interaction Hamiltonian in the probe basis yields
\[ H'_{PS}(x) = X \ox \lp n_3^x(x) Z \rp + Y \ox \lp n_3^y(x) Z \rp + Z \ox \lp n_3^z(x) Z \rp\]
Thus $H_2(x) = n_3^y(x) Z$ and its diagonal values are $y_1(x) = n_3^y(x)$, $y_2(x) = - n_3^y (x)$. In this case, if we return to the Riccati equation~\eqref{eqn:riccati}, we see that we must choose values of $q_0 (x)$ and $q_1(x)$ such that both $n_3^y(x)$ and $-n_3^y(x)$ are solutions. If we add and subtract the Riccati equations for the positive and negative solutions, we get instead the following two equations:
\[ 2 \lp n_3^y (x) \rp^2 = q_0 (x), \hspace{0.25in} \text{and} \hspace{0.25in} \partial_x n_3^y(x) = q_1(x) n_3^y(x), \]
and we see that our solution must be
\[ n_3^y(x) = \exp \lp \int q_1(x) dx \rp, \]
with $q_0(x) = 2 \exp \lp 2 \int q_1(x) dx \rp$. The step operators take the form
\begin{eqnarray*}
	M_{\pm}(x) & = & \frac{\ident}{\sqrt{2}} \mp \frac{\d}{\sqrt{2}} \lp n_3^y(x) Z - c(x) \cos \psi(x) \ident \rp \\
	& & + \frac{i \d}{\sqrt{2}} \lp \lp \pm n_3^x(x) + n_3^z(x) \rp Z \pm c(x) \sin \psi(x) \ident \rp
\end{eqnarray*}
where, as we've defined before, $c(x) e^{ i \psi(x)}$ is the warping of the probe basis, and $q_1(x) = - 4 c(x) \cos \psi(x)$. We can simplify this operator by choosing $\psi(x) = 0$ and $n_3^z(x) = 0$. This forces $n_3^z(x) = \sqrt{1 - \lp n_3^y(x) \rp^2}$. The simplified step operator is
\begin{equation}
	M_{\pm}(x) = \frac{\ident}{\sqrt{2}} \mp \frac{\d}{\sqrt{2}} \lp e^{-4 \int c(x) dx} Z - c(x) \ident \rp \pm \frac{i \d n_3^x(x) Z}{\sqrt{2}}.
\end{equation}
This gives the first endpoint operator the following form
\begin{eqnarray*}
	M_1 & \propto & W_1 \cdot \text{diag} \lb \exp \lp \int_0^X \lp -e^{-4 \int c(x) dx} + c(x) \rp dx \rp, \right. \\
	& & \hspace{0.69in}  \left. \exp \lp \int_0^X \lp e^{-4 \int c(x) dx} + c(x) \rp dx \rp \rb,
\end{eqnarray*}
and the second, the form
\begin{eqnarray*}
	M_2 & \propto & W_2 \cdot \text{diag} \lb \exp \lp \int_{-X}^0 \lp e^{-4 \int c(x) dx} - c(x) \rp dx \rp, \right. \\
	& & \hspace{0.69in} \left. \exp \lp \int_{-X}^0 \lp -e^{-4 \int c(x) dx} - c(x) \rp dx \rp \rb .
\end{eqnarray*}
If we choose the probe basis warping $c(x)$ to be
\[ c(x) = \frac{1}{2} \lp \tanh \lp x-a\rp + \tanh \lp x-b \rp \rp, \]
then we recover the endpoint operators
\[ M_1 \propto W_1 \cdot \diag{e^{\int_0^X \tanh(x-a) dx}, e^{\int_0^X \tanh(x-b) dx} }, \]
\[ M_2 \propto W_2 \cdot \diag{e^{- \int_{-X}^0 \tanh(x-a) dx}, e^{- \int_{-X}^0 \tanh(x-b) dx} }, \]
where $W_1 = e^{i \theta Z}$ and $W_2 = e^{-i \theta Z}$ for some value $\theta$ resulting from the integration of $n_3^x(x)$. An appropriate choice of $a$ and $b$ will yield the desired generalized diagonal measurement:
\[a = \ln \sqrt{\frac{\tanh X + (2 \a - 1)}{\tanh X - (2 \a -1)}}, \hspace{0.15in} b = \ln \sqrt{\frac{\tanh X + (2 \b - 1)}{\tanh X - (2 \b -1)}}. \]
The following choice of probe basis corresponds to the values of $\vec{n}_3$ described above:
\[ \ls \begin{array}{c|c|c} & & \\ \vec{n}_1 & \vec{n}_2 & \vec{n}_3 \\ & & \end{array} \rs = \ls \begin{array}{ccc} 0&-c(x)&\sqrt{1-c(x)^2} \\ 0 & \sqrt{1-c(x)^2}& c(x) \\ 1&0&0 \end{array} \rs \]
Figure~\ref{fig:ZZ} shows a simulation of this scheme for $\a = 0.8$, $\b=0.2$ where the initial state of the system qubit is $\ket{\Psi}=\ket{+}$.
\begin{figure}
	\begin{subfigure}[b]{\linewidth}
		\fbox{\includegraphics[scale=0.9, clip=true, trim= 175 270 175 310]{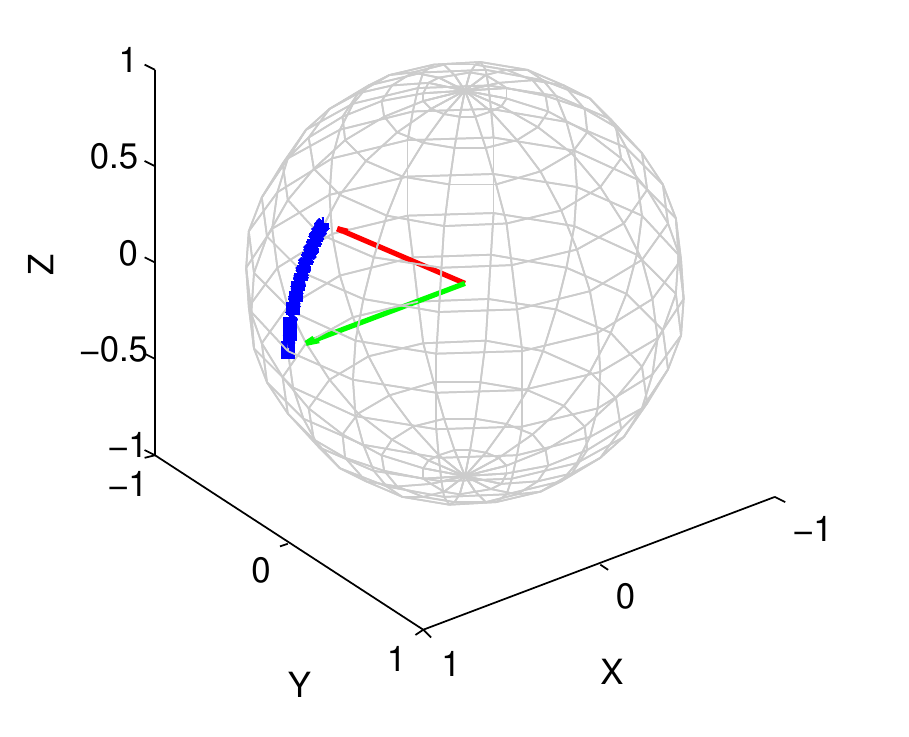}}
		\caption{}
	\end{subfigure}
	
	\begin{subfigure}[b]{0.49\linewidth}
		\fbox{\includegraphics[scale=0.45, clip=true, trim= 170 265 195 290]{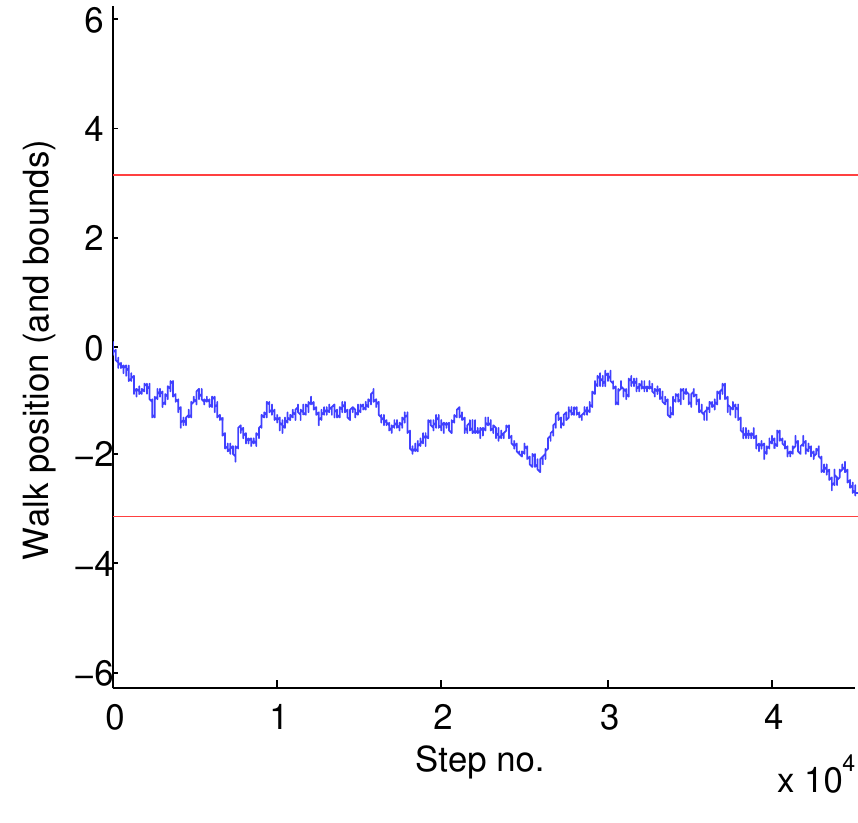}}
		\caption{}
	\end{subfigure}
	\begin{subfigure}[b]{0.49\linewidth}
		\fbox{\includegraphics[scale=0.45, clip=true, trim= 170 265 195 290]{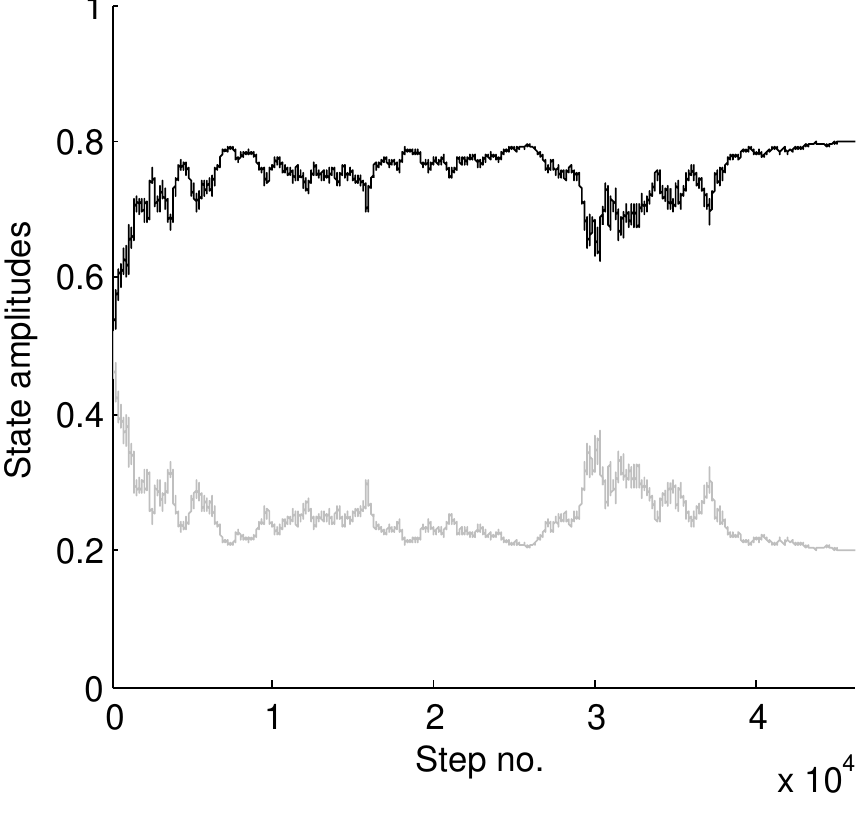}}
		\caption{}
	\end{subfigure}
	\caption{\emph{(color online)} (a) At the beginning of the process the system qubit is in the $\ket{+}$ state, indicated by the green vector. The continuous measurement procedure causes the state to walk along the blue curve on the surface of the sphere, sometimes reversing direction and doubling back along it. At the end of the process, the system qubit reaches the state $M_2 \ket{+} / p_2$. (b) The random walk undertaken by the pointer variable $x$, illustrated by the blue line, ends when the value of $x$ reaches either of the boundaries illustrated by the red lines. (c) The amplitudes of the state evolve towards their post-measurement values. }
	\label{fig:ZZ}
\end{figure}

\section{\label{sec:conclusions}Conclusions}
In this work, we've shown that a continuous decomposition of a two-outcome quantum measurement using a probe qubit and a constant interaction Hamiltonian can only yield measurements with two distinct singular values. In the qubit-to-qubit scheme of section~\ref{sec:qubit} this corresponds to a biased diagonal measurement of the system qubit. Of course, if we consider the recipe outlined in~\cite{weakuniversal} and use our decomposition of two-outcome measurements in sequence to give rise to an $n$-outcome measurement, then this larger measurement procedure can contain operators with $n$ distinct singular values.

The restriction to two singular values is a direct consequence of the reversibility condition (Eq.~\eqref{eqn:reversibility}). This condition, however, is a critical piece of the construction as it guarantees that the desired ``strong'' quantum measurement is faithfully produced at the endpoints. Without the reversibility condition, the continuous measurement procedure is not guaranteed to halt.

In some sense, our scheme is a restricted version of the large-ancilla continuous measurements in~\cite{nonlocality}. Although high-dimensional unitary rotations are not a limiting requirement for quantum computers, for individual quantum systems they can still be restrictive. Our scheme reduces the requirement on the number of probes that need to be simultaneously entangled with the system to one probe for a brief interaction time. It also characterizes the possible measurements in terms of the interaction Hamiltonian of the quantum system. The scheme presented here could be of use for generalized measurements in some types of qubits, such as superconducting qubits~\cite{superconductingnotes}. 

Although we've only analyzed \emph{probe} feedback here, we can also analyze \emph{Hamiltonian} feedback, where control parameters of the Hamiltonian become functions of the pointer variable $x$. In this case, the reversibility condition restricts not only the probes and detector states but also the values assigned to the controls in the interaction Hamiltonian. A detailed analysis of this scheme is forthcoming.

Finally, one could also extend this continuous feedback scheme beyond qubit probes to qudit probes or even continuous probe states. One natural reason for doing so would be to embed the control parameters of a Hamiltonian, as described above, into free parameters of the probe state. Alternatively, schemes such as~\cite{generalizedstochastic} provide a framework for decomposing $n$-outcome measurements into a single stochastic process and this framework could be extended to the interacting probe case using higher-dimensional probes.

\begin{acknowledgments}
JF and TAB thank Daniel Lidar and Ognyan Oreshkov for useful discussions.  This research was supported in part by the ARO MURI under Grant No. W911NF-11-1-0268.
\end{acknowledgments}

\bibliographystyle{apsrev4-1}
\bibliography{QuantumControl}

\appendix
\section{\label{sec:miscellaneouscalculations}Calculation of the $O(\d^2)$ terms of the reversibilty condition}
Starting with Eq.~\eqref{eqn:O2toappendix}, we can write $O(\d^2)$ in three parts, first
\[ M^{(2)}_{\mp}+ M^{(2)}_{\pm} = \sqrt{2} \langle (H'_{PS})^2 \rangle_0 - 2\sqrt{2}cie^{i \psi} \bra{0} H'_{PS} \ket{1}. \]
Next,
\begin{eqnarray*}
	\lefteqn{M^{(1)}_{\mp}M^{(1)}_{\pm} = - \bra{\mp}H'_{PS} \ket{0} \bra{\pm}H'_{PS} \ket{0}} \\
	& & \pm \frac{cie^{i \psi}}{\sqrt{2}} \lp \bra{\pm}H'_{PS} \ket{0} - \bra{\mp}H'_{PS} \ket{0} \rp - \frac{c^2e^{2 i \psi} \ident}{2} \\
	& = & - \bra{\mp}H'_{PS} \ket{0} \bra{\pm}H'_{PS} \ket{0} - cie^{i \psi} \bra{1}H'_{PS} \ket{0} - \frac{c^2e^{2 i \psi} \ident}{2} .
\end{eqnarray*}
We can always discard terms proportional to $\ident$. Grouping together the last two calculations, we get
\begin{eqnarray*}
	\lefteqn{M^{(2)}_{\mp}+ M^{(2)}_{\pm} - 2\sqrt{2}M^{(1)}_{\mp}M^{(1)}_{\pm}} \\ 
	& = & \sqrt{2} \langle (H'_{PS})^2 \rangle_0 - 2\sqrt{2}cie^{i \psi} \lp \bra{0} H'_{PS} \ket{1} - \bra{1} H'_{PS} \ket{0} \rp \\
	& & + 2\sqrt{2} \bra{\mp}H'_{PS} \ket{0} \bra{\pm}H'_{PS} \ket{0} \\
	& = & \sqrt{2} \langle (H'_{PS})^2 \rangle_0 + 4\sqrt{2}ce^{i \psi} H_2 \\
	& & + 2\sqrt{2} \bra{\mp}H'_{PS} \ket{0} \bra{\pm}H'_{PS} \ket{0}.
\end{eqnarray*}
The two terms above, still expressed as functions of $H'_{PS}$, can be expanded as follows. First, note that
\begin{eqnarray*}
	(H'_{PS})^2 & = & \ident \ox \lp H_1^2 + H_2^2 + H_3^2 \rp + X \ox i \ls H_2, H_3 \rs \\
	& & + Y \ox i \ls H_3, H_1 \rs + Z \ox i \ls H_1, H_2 \rs,
\end{eqnarray*}
implying that
\[ \langle (H'_{PS})^2 \rangle_0 = H_1^2 + H_2^2 + H_3^2 + i \ls H_1, H_2 \rs . \]
Next note that
\[ \bra{\pm} H'_{PS} \ket{0} = \frac{1}{\sqrt{2}} \lp \pm H_1 \pm i H_2 + H_3 \rp, \]
which yields
\begin{eqnarray*}
     \lefteqn{\bra{\mp} H'_{PS}\ket{0} \bra{\pm} H'_{PS} \ket{0}} \\
    & = &-\frac{1}{2} \lp H_1 + i H_2 \mp H_3 \rp\lp H_1 + i H_2 \pm H_3 \rp \\
    & = & \frac{1}{2} \lp - \lp H_1 + iH_2 \rp^2 + H_3^2 \pm \ls H_3, H_1 \rs \pm i \ls H_3, H_2 \rs \rp .
\end{eqnarray*}
Putting this all together, we can now write Eq.~\eqref{eqn:d2terms}.

\section{Various lemmas}
\setcounter{lem}{0}
\begin{lem}[Probe basis of a weak measurement] \label{lem:probebasis}
Any diffusive weak measurement given by a probe feedback loop (Def.~\ref{def:probefeedbackloop}) with a probe basis $\{\vec{n}_1(x), \vec{n}_2(x), \vec{n}_3(x)\}$ must have $\vec{n}_2 \cdot \vec{\s} \sim O(\d)$. Thus, an orthonormal basis for the Bloch sphere that approximates the probe and detector always exists.
\end{lem}
\begin{proof}
Recall that for a probe feedback loop we can expand the operator in orders of $\d$,
\begin{eqnarray*}
	M_{\pm}(x) & = & \bra{\Psi^{\pm}} e^{i \d H_{PS}} \ket{\s} \\
	& = & \bra{\Psi^{\pm}} \lp \ident + i \d H_{PS} \rp \ket{\s} \\
	& \approx & \braket{\Psi^{\pm}}{\s} \cdot \ident + O(\d)
\end{eqnarray*}
A diffusive weak measurement must always obtain both results $\pm$ with nearly equal probability (up to $O(\d)$). The probability of each result on a quantum state $\ket{\Phi}$ is
\begin{eqnarray*}
	p_{\pm} & = & \bra{\Phi} M^{\dag}_{\pm}(x) M_{\pm}(x) \ket{\Phi} \\
	& \approx & |\braket{\Psi^{\pm}}{\s} |^2  \cdot \braket{\Phi}{\Phi} + O(\d) \\
	& = & \frac{1}{2} + O(\d)
\end{eqnarray*}
which in turn means that $|\braket{\Psi^{\pm}}{\s} | \approx 1/\sqrt{2} + O(\d)$. In the Bloch vector representation, this implies that $\vec{n}_2 \cdot \vec{\s} \sim O(\d)$.
\end{proof}

\begin{lem}[Commutator identity]
\label{lem:commutatoridentity}
For any Hermitian operators $\hat{O}$ and $\hat{A}$, if $i \ls \hat{O}, \hat{A} \rs \propto \hat{O}$ then $\hat{O}=0$.
\end{lem}
\begin{proof}
First, let us express $\hat{A}$ and $\hat{O}$ in a basis where $\hat{A}$ is diagonal, i.e.: $\hat{A} = \sum_i a_i \ketbra{i}$ with $a_i \in \mbb{R}$ and $\hat{O} = \sum_{jk} o_{jk} \ketb{j}{k}$. This makes our equation
\[ i \sum_{jk} o_{jk} \ketb{j}{k} \hat{A} - i \hat{A} \sum_{jk} o_{jk} \ketb{j}{k} = \a \sum_{jk} o_{jk} \ketb{j}{k} .\]
for some constant $\a$. Expanding the operator $\hat{A}$ yields
\[ i \sum_{ijk} \lp a_io_{jk} \ket{j} \braket{k}{i} \bra{i} - a_i o_{jk} \ket{i} \braket{i}{j} \bra{k} \rp = \a \sum_{jk} o_{jk} \ketb{j}{k}, \]
and this, in turn, reduces to
\[ i \sum_{jk} \lp a_ko_{jk} \ketb{j}{k} - a_j o_{jk} \ketb{j}{k} \rp = \a \sum_{jk} o_{jk} \ketb{j}{k} . \]
This implies that for all $j,k$ we have $(a_k-a_j) \cdot i o_{jk} = \a o_{jk}$ and we find that $\Re \lb o_{jk} \rb = - \a (a_k - a_j) \cdot \Im \lb o_{jk} \rb$ as well as $\Re \lb o_{jk} \rb =  \a \Im \lb o_{jk} \rb /(a_k-a_j)$, leading to a contradiction. The only valid solution remaining is $o_{jk}=0$ for all $j,k$ exactly.
\end{proof}

\end{document}